\documentclass[11pt, letterpaper]{article}

\usepackage{amsmath,amssymb,amsthm,amsfonts}

\usepackage[margin=2.5cm]{geometry}
\usepackage[colorlinks,citecolor=blue,bookmarks=true, draft=false]{hyperref}
\usepackage[lined,boxed, vlined, ruled]{algorithm2e}

\theoremstyle{plain}
  \newtheorem{lem}{\protect\lemmaname}
\theoremstyle{plain}
\newtheorem{thm}{\protect\theoremname}
\theoremstyle{plain}
  
  \theoremstyle{remark}
  \newtheorem{rem}{\protect\remarkname}

\makeatother
\providecommand{\lemmaname}{Lemma}
\providecommand{\remarkname}{Remark}
\providecommand{\corollaryname}{Corollary}
\providecommand{\theoremname}{Theorem}

\begin{document}

\global\long\def\ALG{\text{\normalfont ALG}}
\global\long\def\OPT{\text{\normalfont OPT}}
\global\long\def\PACK{\text{\normalfont P}}
\global\long\def\E{\text{\normalfont E}}
\global\long\def\supp{\text{\normalfont supp}}

\title{Online Multidimensional Packing Problems in the Random-Order Model}

\author{
David Naori\thanks{Computer Science Department, Technion, Israel. Email:\texttt{\{dnaori,danny\}@cs.technion.ac.il}.} 
\and Danny Raz\footnotemark[1]
}
\date{}

\maketitle

\begin{abstract}
    We study online multidimensional variants of the generalized assignment problem which are used to model prominent real-world applications, such as the assignment of virtual machines with multiple resource requirements to physical infrastructure in cloud computing.
    These problems can be seen as an extension of the well known secretary problem and thus the standard online worst-case model cannot provide any performance guarantee. The prevailing model in this case is
    the random-order model, which provides a useful realistic and robust alternative. Using this model, we study the $d$-dimensional generalized assignment problem, where we introduce a novel technique that achieves an $O(d)$-competitive algorithms and prove a matching lower bound of $\Omega(d)$. Furthermore, our algorithm improves upon the best-known competitive-ratio for the online (one-dimensional) generalized assignment problem and the online knapsack problem. 
\end{abstract}

\newpage

\section{Introduction}
Online multidimensional packing problems appear in a wide verity of real-world applications. A recent relevant example is the assignment of virtual elements to the physical infrastructure in Network Function Virtualization (NFV) and cloud computing. Typically, in these problems, we are given a set of bins, each with a certain capacity profile, then, items arrive one-by-one in an online fashion, each with a certain size and profit. Upon each arrival, one has to decide immediately and irrevocably whether and where to pack the current item. The goal is to find an assignment that maximizes the total profit without exceeding the capacity of any bin. These problems can be viewed as generalizations of the well-known secretary problem, in which we have a single bin, and every secretary consumes the capacity of the whole bin.

The common way of analyzing online algorithms is to use the worst-case model,
where an adversary picks an instance along with the order in which items are revealed to the online algorithm.
Despite its prevalence in the analysis of online algorithms, this setting is too pessimistic for the problem at hand. Indeed, no online algorithm can achieve any non-trivial worst-case competitive-ratio, even for the simple case of the secretary problem, as shown by Aggarwal et al.~\cite{DBLP:conf/soda/AggarwalGKM11}. A more realistic model is the random-order model 
in which the power of choosing the arrival order of items is taken away from the adversary, instead, 
the arrival order is chosen uniformly at random.
In this model, we say that an algorithm $\ALG$ is $c$-competitive if for every input instance $\mathcal{I}$ it holds that,
$ c \cdot \E[\ALG(\mathcal{I})] \geq   \OPT(\mathcal{I})$, where the expectation is taken over the random arrival orders and the randomness of the algorithm.\footnote{We follow the definition used in
\cite{DBLP:conf/approx/BabaioffIKK07, DBLP:conf/soda/BabaioffIK07, DBLP:conf/esa/KesselheimRTV13} although it is also common to refer to such algorithm as $1/c$-competitive.}

Kesselheim et al.~\cite{DBLP:conf/esa/KesselheimRTV13, DBLP:journals/siamcomp/KesselheimRTV18} generalized the known optimal algorithm for the secretary problem to various packing problems in the random-order model. The outline of the generic algorithmic framework is as follows: it starts with a sampling phase in which the algorithm only observes the arriving items. Then, at every subsequent online round, the algorithm computes a local solution for the sub-instance consists of all the items that arrived so far. If the bin in which the current item is packed in this local solution has enough free capacity (i.e., an assignment of the current item in this bin is feasible) the algorithm carries it out, otherwise, it leaves the item unpacked.
Using this framework, 
Kesselheim et al.~\cite{DBLP:journals/siamcomp/KesselheimRTV18} presented an algorithm for the online generalized assignment problem (GAP) with the best-known competitive-ratio (prior to this work).

In GAP we have a set of bins and a set of items.  Each bin has a certain non-negative capacity and each item has several packing options, one for each bin. Each packing option is associated with a certain consumption from the capacity of the bin and a certain profit it provides. The goal is to
pack the items in the bins where each item can be packed at most once, maximizing the total profit without exceeding the capacity of any bin.
A major challenge in online packing problems, and online GAP in particular, is to handle both
items with high consumption of resources compared to a bin capacity, as well as items with low consumption.
Kesselheim et al. handle this challenge by partitioning a GAP instance into two sub-instances: the first contains all ``heavy'' packing options of items, that is, packing options that occupy more than half of a bin capacity, the second is the complementary sub-instance that contains all ``light'' packing options. Their algorithm makes an initial random choice to operate on one of the sub-instances exclusively. Although it achieves the best-known results, this behaviour is undesirable for most applications, since it always leaves one type of items unpacked.

We use a similar algorithmic framework to design an online algorithm for the $d$-dimensional generalization of online GAP, or online Vector Generalized Assignment Problem (VGAP), in which the capacity profile of each bin, as well as the consumption of items from each bin, is described by a $d$-dimensional vector. The goal remains to maximize the profit, while the capacity of each bin must not be exceeded in any of its $d$ dimensions. To the best of our knowledge, this is the first time the online version of this problem is studied. Our algorithm offers a preferable behaviour and improves upon the best-known competitive-ratio for online GAP. To achieve this, we take a different approach to overcome the challenge: instead of limiting the algorithm to either ``heavy'' or ``light'' packing options, our algorithm considers them both. It operates in three phases: a sampling phase, a phase for ``heavy'' packing options, and a phase for ``light'' packing options. To compute the tentative assignments, our algorithm in the second phase uses maximum-weight bipartite matching, and in the third phase, it uses an optimal fractional solution for the LP-relaxation of the local problem, and randomized rounding. 

We also apply our technique to the $\{0,1\}$-VGAP in which every packing option of an item in every dimension must consume either the whole capacity of the bin or non of it. In one-dimension this problem is identical to weighted bipartite matching. For $\{0,1\}$-VGAP we partition the instance by a different criterion: the number of non-zero entries in the consumption vector of a packing option.



Another interesting special case of VGAP is the Vector Multiple Knapsack Problem (VMKP), in which all bins are identical, and the packing options of each item are identical for all bins. That is, regardless of the bin's identity, the item consumes the same amount of capacity and raises the same profit. For instances of VMKP with at least two bins, we describe a simpler algorithm that avoids partitioning the instance. Here, our algorithm uses a fractional solution for the LP-relaxation of the local problem only to make a binary decision whether to pack the current item or not. For the actual packing, it exploits the fact that all packing options are identical and uses greedy First Fit approach, typically used for the Bin Packing problem.

Finally, we prove a lower bound for the online vector knapsack problem in the random-order model, which also applies to VMKP and VGAP, and indicates that our algorithms are asymptotically optimal. This lower bound is inspired by the work of Babaioff et al.~\cite{DBLP:conf/soda/BabaioffIK07} on the matroid secretary problem, 
which is based solely on the inherent uncertainty due to the online nature of the problem without any complexity assumptions.

\paragraph{Our main contributions are:}
\begin{enumerate}
    \item We describe an algorithm for online VGAP with a competitive-ratio of $\sqrt[4]{e}\left(4d + 2\right) \approx 5.14d+2.57$, where $d$ is the dimension. For the VMKP with at least two bins we describe a $\left(4d + 2 \right)$-competitive algorithm. To the best of our knowledge, these problems are studied for the first time.
    
    \item We prove a matching lower bound of $\Omega(d)$ which is valid both  for VGAP and VMKP.  
    
    \item Our method improves upon the best-known competitive-ratio for (one-dimensional) GAP from $8.1$ to $6.99$ (which is also the best-known competitive-ratio for online knapsack).
\end{enumerate}

\section{Related Work}
Online packing problems in the random-order model have been studied extensively in recent years,
most of them are generalizations of the secretary problem which has an optimal $e$-competitive algorithm~\cite{dynkin1963optimum, lindley1961dynamic}. An immediate generalization is the multiple-choice secretary problem, in which one is allowed to pick up to $k$ secretaries. It was studied by Kleinberg~\cite{DBLP:conf/soda/Kleinberg05}, where he presented an asymptotically optimal $\frac{\sqrt{k}}{\sqrt{k} - 5}$-competitive algorithm.
Another related problem is the weighted-matching problem which has an optimal $e$-competitive algorithm by Kesselheim et al.~\cite{DBLP:conf/esa/KesselheimRTV13}.

The online knapsack problem, which generalizes the multiple-secretary problem, was studied by Babaioff et al.~\cite{DBLP:conf/approx/BabaioffIKK07} who presented an $10 e$-competitive algorithm. It was later improved by the work of Kesselheim et al.~\cite{DBLP:journals/siamcomp/KesselheimRTV18} on online GAP, which generalizes all of the above problems. They presented an $8.1$-competitive algorithm which is the best-known competitive-ratio for online GAP and the online knapsack problem. Our result for VGAP improves on that.

In their work, Kesselheim et al. also studied the online packing LPs problem with column sparsity $d$. The general online packing LPs problem was studied before by~\cite{DBLP:journals/corr/Agrawal,DBLP:conf/esa/FeldmanHKMS10,DBLP:conf/icalp/MolinaroR12}.
In this problem, there is a set of resources and a set of requests. Each request has several options to be served and each option is associated with a profit and a certain demand from each resource. For column sparsity $d$, each request may have a demand from at most $d$ of the resources. This problem generalizes VGAP studied in this paper, however, to the best of our knowledge, the only known competitive online algorithms for this problem are for the special case of $B \geq 2$, where $B$ is the capacity ratio, i.e., the minimal ratio between the capacity of a resource and the maximum demand for this resource. For this case they presented an $O\left(d^{1/(B-1)}\right)$-competitive algorithm which in case $B=\Omega\left(\log{d}/\epsilon^{2}\right)$ is $\left(1+\epsilon\right)$-competitive.

Dean et al.~\cite{DBLP:conf/soda/DeanGV05} showed that under the assumption of $\text{NP} \neq \text{ZPP}$, the packing integer programs problem (PIP, also known as vector knapsack) which is a special case of VMKP, cannot be approximated in polynomial time to within $d^{1-\epsilon}$ for any $\epsilon > 0$ even in the offline settings. Under the same assumptions, Chekuri et al.~\cite{DBLP:conf/soda/ChekuriK99} showed that the  $\{0,1\}$-case cannot be approximated to within $d^{1/2-\epsilon}$ for any $\epsilon > 0$. Their results are also applicable for the offline VMKP, VGAP and the $\{0,1\}$-VGAP. By using the results of Zuckerman~\cite{DBLP:journals/toc/Zuckerman07}, the same hardness result can be proved under the weaker assumption of $\text{P} \neq \text{NP} $ instead. As opposed to these results, our lower bound holds with no complexity assumptions, and applies even for algorithms with unbounded computational power.

Some related problems have competitive algorithms in the worst-case model too. One example is the AdWords problem which is a special case of GAP in which the profit of each item is equal to its size. Under the assumption that items are small compared to the capacity of the bins, Metha et al.~\cite{DBLP:journals/jacm/MehtaSVV07} presented an optimal $\frac{e}{e-1}$-competitive algorithm. Without this assumption, the best known competitive-ratio is $2$~\cite{DBLP:journals/geb/LehmannLN06}.
Another example is the online vector bin packing problem, in which items arrive one-by-one, and the goal is to pack them all in the minimum number of unit sized $d$-dimensional bins. This problem was studied by Garey et al.~\cite{garey1976resource} who showed that the First Fit algorithm has a worst-case competitive-ratio of $(d+0.7)$. More recently, Azar et al.~\cite{DBLP:conf/stoc/AzarCKS13} showed that this algorithm is asymptotically optimal by proving a lower bound of $\Omega\left(d^{1-\epsilon}\right)$. 


\section{Vector Generalized Assignment Problem} \label{VGAP section}

In the $d$-dimensional \emph{Generalized Assignment Problem} (VGAP),
we have a set of $m$ $d$-dimensional bins and a set of $n$
$d$-dimensional items that may be packed in the bins. Each bin $j$ has a capacity
$\mathbf{b}_{j}=\left(b_{j}^{1},\dots,b_{j}^{d}\right) \in \mathbb{R}_{\geq 0}^{d}$.
Packing item $i$ in bin $j$ consumes an amount of 
$\mathbf{w}_{i,j}=\left(w_{i,j}^{1},\dots,w_{i,j}^{d}\right)\in\mathbb{R}_{\geq0}^{d}$
from bin's $j$ capacity and provides a profit of $p_{i,j}\geq0$.
Each item may be packed in at most one of the bins and the capacity of each bin must not be exceeded in any of its $d$ dimensions.
The goal is to find a feasible packing that maximizes the total profit. We use the following LP-formulation:

\begin{center}
\begin{tabular}{ccc}
max & ${\displaystyle \sum_{i\in\left[n\right],\  j\in\left[m\right]}}p_{i,j}x_{i,j}$ & \\
 &  & \\
s.t. & ${\displaystyle \sum_{i\in\left[n\right]}}w_{i,j}^{t}x_{i,j}\leq b_{j}^{t}$ & \ \  \ \ \ \ \ \ \ \ \ \  $j\in\left[m\right],t\in\left[d\right]$\\
 &  & \\
 & ${\displaystyle \sum_{j\in\left[m\right]}}x_{i,j}\leq1$ &  \ \ \ \ \ \ \ \ \ \ \ \ \ \ \ \ \ \ \ \ \ \ \ $i\in\left[n\right]$\\
 &  & \\
 & $x_{i,j}\in\left\{ 0,1\right\} $ & \ \  \ \ \ \ \ \ \ \ \  \ $i\in\left[n\right],j\in\left[m\right]$.
\end{tabular}
\end{center}

We consider the online version of the problem in which the set of
bins and their capacities are initially known, as well as the total
number of items $n$. The items, however, arrive one by one in a random
order. When item $i$ arrives, we learn its \emph{packing options}, i.e.,
its consumption on every bin $\mathbf{w}_{i,1},\dots,\mathbf{w}_{i,m}$ (which we also call the \emph{weight vectors} of $i$)
along with the corresponding profits $p_{i,1},\dots,p_{i,m}$. After
every arrival, an immediate and irrevocable decision must be made:
Assign the item to one of the available bins or leave the item unpacked.

Our algorithm is based on the technique presented
by the authors of~\cite{DBLP:journals/siamcomp/KesselheimRTV18} with several critical improvements (see Algorithm~\ref{VGap Algorithm}). We call the packing
option of item $i$ in bin $j$ \emph{light }if $w_{i,j}^{t} \leq b_{j}^{t}/2$, 
$\forall t\in\left[d\right]$, otherwise, we call it \emph{heavy}.
Given a GAP instance $\mathcal{I}$ we partition it into two sub-instances
$\mathcal{I}_{heavy}$ and $\mathcal{I}_{light}$, both consist of
the original items and bins, however, $\mathcal{I}_{heavy}$ consists
only of the heavy packing options of every item, while $\mathcal{I}_{light}$
consists only of the light ones. In contrast to the algorithm presented
in~\cite{DBLP:journals/siamcomp/KesselheimRTV18} that makes a random choice whether
to operate on $\mathcal{I}_{heavy}$ or $\mathcal{I}_{light}$ exclusively,
our algorithm considers them both. It is based on the intuition that
heavy options may need a chance to be packed first, since any other
packing decision might prevent them from being packed, while light
options are more likely to fit in. Our algorithm operates in three
phases: the \emph{sampling phase} in which it only observes the arriving
items, the \emph{heavy phase} in which it considers only heavy options,
and the \emph{light phase} in which it considers only light options.
In the heavy phase, our algorithm uses a matching in a weighted bipartite
graph to make packing decisions, to this end, given an instance $\mathcal{I}$
we define a weighted bipartite graph $G\left(\mathcal{I}\right)=\left(L,R,E\right)$, where $L$ is the set of items of $\mathcal{I}$, $R$ is the
set of bins of $\mathcal{I}$, and there exists an edge $\left(i,j\right)\in E$ 
of weight $p_{i,j}$ if item $i$ can be packed in bin $j$ (i.e., $w_{i,j}^{t}\leq b_{j}^{t}$, $\forall t\in\left[d\right]$).
Each phase takes place in a 
continuous fraction of the online rounds.
To partition the rounds into phases, we use two parameters $q_{1}$
and $q_{2}$ that will be defined thereafter. For convenience of presentation
and analysis, we represent a packing by a set $\PACK \subseteq \left[n\right] \times \left[m\right]$
such that $\PACK=\left\{ \left(i,j\right) : i\text{ is packed in bin \ensuremath{j}}\right\} $.
We also define $p_{i,0}=0$, $\forall i\in\left[n\right]$. For an
instance $\mathcal{I}$ and a subset $S$ of its items, we denote
by $\mathcal{I}|_{S}$ the sub-instance that consists only of the
items in $S$. 


\IncMargin{1em}
\begin{algorithm}
\label{VGap Algorithm}
\caption{Online VGAP}

$S{}_{0}\leftarrow\emptyset$, $\PACK_{0}\leftarrow\emptyset$\;

\For { each item $i_{\ell}$ that arrives at round $\ell$ } {
    $ S_{\ell} \leftarrow S_{\ell-1} \cup \left\{ i_{\ell} \right\} $\;
    \uIf (\tcc*[f]{sampling phase}) {$\ell\leq q_{1}n$} {
		continue to the next round\;
    }
    \uElseIf( \tcc*[f]{heavy phase}){ $q_{1}n+1\le\ell\leq q_{2}n$ } {
        Let $x^{\left(\ell\right)}$ be a maximum-weight matching in
        $G\left(\mathcal{I}_{heavy}|_{S_{\ell}}\right)$;

		\tcp*[h]{compute a tentative assignment $\left(i_{\ell},j_{\ell}\right)$}

		\uIf {$i_\ell$ is matched in $x^{\left(\ell\right)}$} {
			Let $j_\ell$ be the bin to which $i_\ell$ is matched\; 
		} 
		\Else {
			$j_\ell \leftarrow 0 $
		}
		\If {$j_\ell\neq0$ {\bf and} $j_\ell$ is empty in $\PACK_{\ell - 1}$} {
			$\PACK_{\ell}\leftarrow \PACK_{\ell-1}\cup\left\{ \left(i_{\ell},j_{\ell}\right)\right\}$\;
    	}
    }
    \Else( \tcp*[h] {({$\ell \geq q_{2}n + 1$})} \tcc*[f]{ light phase  }) { 
        Let $x^{\left(\ell\right)}$ be an optimal fractional solution for
        the LP-relaxation of $\mathcal{I}_{light}|_{S_{\ell}}$\;

		\tcp*[h]{compute a tentative assignment $\left(i_{\ell},j_{\ell}\right)$ by randomized rounding}

		Choose bin $j_{\ell}$ randomly where $\Pr\left[j_{\ell}=j\right]=x_{i_{\ell},j}^{\left(\ell\right)}$
    	and $\Pr\left[j_{\ell}=0\right]=1-{\displaystyle \sum_{j\in\left[m\right]}}x_{i_{\ell},j}^{\left(\ell\right)}$\;
   	 \If{$j_\ell\neq0$ {\bf and} $ \PACK_{\ell - 1} \cup \left\{\left(i_{\ell},j_{\ell}\right)\right\}$ is feasible} {
        	$\PACK_{\ell}\leftarrow \PACK_{\ell-1}\cup\left\{ \left(i_{\ell},j_{\ell}\right)\right\}$\;
    	}
    }
}
\Return{ $\PACK_{n}$ }
\end{algorithm}
\DecMargin{1em}

We now analyze the performance of Algorithm~\ref{VGap Algorithm}.
Let $\OPT\left(\mathcal{I}\right)$ and $\ALG\left(\mathcal{I}\right)$
denote the overall profit of the optimal packing and the overall profit
of the packing produced by Algorithm~\ref{VGap Algorithm} on instance
$\mathcal{I}$ respectively.
Let $R_{\ell}$ denote the profit raised by the algorithm at round
$\ell$. In Lemma~\ref{Heavy phase lemma} and Lemma~\ref{Light phase lemma}
below, we bound the expected profit raised at each round of the heavy
phase and the light phase respectively. Similar claims are presented
in~\cite{DBLP:conf/esa/KesselheimRTV13} and~\cite{DBLP:journals/siamcomp/KesselheimRTV18}.

\begin{lem}
\label{Heavy phase lemma}For $q_{1}n+1\le\ell\leq q_{2}n$, we have
\[
\E\left[R_{\ell}\right]\geq\frac{q_{1}}{\ell-1}\cdot\frac{1}{d}\OPT\left(\mathcal{I}_{heavy}\right).
\]
\end{lem}

\begin{proof}
Let $x^{*}$ be an optimal solution for $\mathcal{I}_{heavy}$, hence,
$p^{T}x^{*}=\OPT\left(\mathcal{I}_{heavy}\right)$, and let $x^{*}|_{S_{\ell}}$
denote the projection of $x^{*}$ onto the set of items $S_{\ell}$,
i.e., $\left(x^{*}|_{S_{\ell}}\right)_{i,j}=x_{i,j}^{*}$ if $i\in S_{\ell}$
and $\left(x^{*}|_{S_{\ell}}\right)_{i,j}=0$ otherwise. Observe (by the definition of heavy) that
in $x^{*}|_{S_{\ell}}$ every bin holds at most $d$ items. Let $x_{\ell}^{*}$
be the solution obtained from $x^{*}|_{S_{\ell}}$ by leaving only
the most profitable item in each bin. We get $p^{T}x_{\ell}^{*}\geq\frac{1}{d}\cdot p^{T}\left(x^{*}|_{S_{\ell}}\right)$.
Also, since $x_{\ell}^{*}$ is a feasible matching in $G\left(\mathcal{I}_{heavy}|_{S_{\ell}}\right)$,
we have $p^{T}x^{\left(\ell\right)}\geq p^{T}x_{\ell}^{*}\geq\frac{1}{d}\cdot p^{T}\left(x^{*}|_{S_{\ell}}\right)$.
Now since $S_{\ell}\subseteq\left[n\right]$ is a uniformly random
subset of size $\ell$, we have $\E\left[p^{T}\left(x^{*}|_{S_{\ell}}\right)\right]=\frac{\ell}{n}\cdot\OPT\left(\mathcal{I}_{heavy}\right)$.
Also, $i_{\ell}$ can be viewed as a uniformly random item of $S_{\ell}$,
and since $x^{\left(\ell\right)}$ is a matching we have $\E\left[p_{i_{\ell},j_{\ell}}\right]=\E\left[\sum_{j\in\left[m\right]}x_{i_{\ell},j}^{\left(\ell\right)}p_{i_{\ell},j}\right]=\frac{1}{\ell}\E\left[p^{T}x^{\left(\ell\right)}\right]$.
Combining the results together, we get
\[
\E\left[p_{i_{\ell},j_{\ell}}\right]=\frac{1}{\ell}\E\left[p^{T}x^{\left(\ell\right)}\right]\geq\frac{1}{\ell}\E\left[\frac{1}{d}\cdot p^{T}\left(x^{*}|_{S_{\ell}}\right)\right]=\frac{1}{n\cdot d}\OPT\left(\mathcal{I}_{heavy}\right).
\]
The above expectation is taken only over the random choice of the
subset $S_{\ell}\subseteq\left[n\right]$ and the random choice of
$i_{\ell}\in S_{\ell}$, while the arrival order of items in previous
rounds is irrelevant. We now bound the probability of successful assignment
over the random arrival order of previous items. The assignment is
successful if no item is packed in $j_{\ell}$ in rounds $q_{1}n,\dots,\ell-1$.
At round $\ell-1$ the algorithm uses a maximum-weight matching in
$G\left(\mathcal{I}_{heavy}|_{S_{\ell-1}}\right)$ to compute a tentative
assignment $\left(i_{\ell-1},j_{\ell-1}\right)$. In that matching
at most one item is matched to $j_{\ell}$. Since $i_{\ell-1}$ is
a uniformly random item of $S_{\ell-1}$, the probability that $i_{\ell-1}$
is matched to $j_{\ell}$ is at most $1/\left(\ell-1\right)$ regardless
of the arrival order of the items in rounds $1,\dots,\ell-2$, hence,
we can treat subsequent events as independent and repeat the argument
inductively from $\ell-1$ to $q_{1}n+1$ to get, 
\[
\Pr\left[\text{successful assignment}\right]\geq\prod_{k=q_{1}n+1}^{\ell-1}\left(1-\frac{1}{k}\right)=\frac{q_{1}n}{\ell-1}.
\]
Combining the expected profit with the probability of successful assignment, we get
\[
\E\left[R_{\ell}\right]\geq\frac{q_{1}}{\ell-1}\cdot\frac{1}{d}\OPT\left(\mathcal{I}_{heavy}\right). \qedhere
\] 
\end{proof}

\begin{lem}
\label{Light phase lemma}
For $\ell\geq q_{2}n+1$, we have
\[
\E\left[R_{\ell}\right]\geq\frac{q_{1}}{q_{2}}\left(1-2d\sum_{k=q_{2}n+1}^{\ell-1}\frac{1}{k}\right)\frac{1}{n}\OPT\left(\mathcal{I}_{light}\right).
\]
\end{lem}

\begin{proof}
Let $x^{*}$ be an optimal solution for $\mathcal{I}_{light}$. At
round $\ell\geq q_{2}n+1$ the algorithm uses randomized rounding
to determine the tentative assignment of $i_{\ell}$ from the fractional
LP-solution $x^{\left(\ell\right)}$, therefore, $\E\left[p_{i_{\ell},j_{\ell}}\right]=\E\left[\sum_{j\in\left[m\right]}x_{i_{\ell},j}^{\left(\ell\right)}p_{i_{\ell},j}\right]$.
Using this observation, we can now follow a similar argument to that in the proof of Lemma~\ref{Heavy phase lemma}
and get that for $\ell\geq q_{2}n+1$, we have
\[
\E\left[p_{i_{\ell},j_{\ell}}\right]=
\frac{1}{\ell}\E{\left[p^{T}x^{\left(\ell\right)}\right]}\geq
\frac{1}{\ell}\E{\left[p^{T}\left(x^{*}|_{S_{\ell}}\right)\right]}=
\frac{1}{n}\OPT\left(\mathcal{I}_{light}\right),
\]
where the expectation is taken only over the random choice of the
subset $S_{\ell}\subseteq\left[n\right]$, the random choice of $i_{\ell}\in S_{\ell}$
and the internal randomness of the algorithm at round $\ell$. Here
too, we bound the probability of successful assignment over the random
arrival order of previous items and the internal randomness of the
algorithm in previous rounds. Let us denote by $c\left(j,t,\ell\right)$
the total consumption of tentative assignments to bin $j$ in dimension
$t$ during the light phase and before round $\ell$. At round $\ell$,
the algorithm considers only light options, therefore, the assignment
of $i_{\ell}$ to $j_{\ell}$ must be successful if the following
conditions hold: (1) no item was packed in $j_{\ell}$ during the
heavy phase, and (2) for every dimension $t\in\left[d\right]$, $c\left(j_{\ell},t,\ell\right)\le b_{j_{\ell}}^{t}/2$.
Let us denote event (1) by $H_{\ell}$, and the events described in
(2) by $L_{\ell}^{t}$ for every dimension $t\in\left[d\right]$.
We now bound $\E\left[c\left(j_{\ell},t,\ell\right)\right]$ for every
$t\in\left[d\right]$. Fix $t\in\left[d\right]$, at round $k<\ell$
of the light phase, the algorithm computes a tentative assignment
based on a fractional optimal solution for the LP-relaxation of $\mathcal{I}_{light}|_{S_{k}}$.
In that solution, the total consumption of bin $j_{\ell}$ in dimension
$t$ is at most $b_{j_{\ell}}^{t}$. Since $i_{k}$ can be viewed
as a uniformly random item of $S_{k}$, the expected consumption of
$i_{k}$ from $j_{\ell}$ in dimension $t$ is at most $b_{j_{\ell}}^{t}/k$,
where the expectation is taken over the choice of $i_{k}\in S_{k}$
and the internal randomness of the algorithm at round $k$. Therefore, it is independent of the arrival order of items in rounds $1,\dots,k-1$, and the internal randomness used in those rounds.
Hence, $\E\left[c\left(j_{\ell},t,\ell\right)\right]\leq\sum_{k=q_{2}n+1}^{\ell-1}b_{j_{\ell}}^{t}/k$.
We have
\[
\Pr\left[\bigwedge_{t=1}^{d}L_{\ell}^{t}\right]=1-\Pr\left[\bigvee_{t=1}^{d}\neg L_{\ell}^{t}\right]\geq1-\sum_{t=1}^{d}\Pr\left[\neg L_{\ell}^{t}\right]\geq1-\sum_{t=1}^{d}\frac{\sum_{k=q_{2}n+1}^{\ell-1}b_{j_{\ell}}^{t}/k}{b_{j_{\ell}}^{t}/2}\geq1-2d\sum_{k=q_{2}n+1}^{\ell-1}\frac{1}{k}.
\]
 The first inequality is due to a union bound, and the second is 
due to Markov's inequality. Since this event is independent of the
arrival order of items in the heavy phase, we can follow the argument
from the proof of the previous lemma and get
\begin{align*}
\Pr\left[\text{successful assignment}\right] & \geq \Pr\left[H_{\ell}\wedge\bigwedge_{t=1}^{d}L_{\ell}^{t}\right] \\
& \geq \prod_{k=q_{1}n+1}^{q_{2}n}\left(1-\frac{1}{k}\right)\left(1-2d\sum_{k=q_{2}n+1}^{\ell-1}\frac{1}{k}\right) \\
& = \frac{q_{1}}{q_{2}}\left(1-2d\sum_{k=q_{2}n+1}^{\ell-1}\frac{1}{k}\right).
\end{align*}
We can now combine the results of the expected profit and the success
probability to get the lemma.
\end{proof}

\begin{thm}
\label{Therorem_gap_competitive} For $q_2 = 2d/\left(2d+1\right)$ and $q_1 = q_2/\sqrt[4]{e}$, Algorithm~\ref{VGap Algorithm} is $\sqrt[4]{e}(4d + 2)$-competitive.
\end{thm}

\begin{proof}
The overall profit of the algorithm can be written as $\E\left[\sum_{\ell=1}^{n}R_{\ell}\right]=\sum_{\ell=1}^{n}\E\left[R_{\ell}\right]$. We sum over the profit raised in each phase separately. For the heavy phase we have
\begin{align*}
\sum_{\ell=q_{1}n+1}^{q_{2}n}\E\left[R_{\ell}\right] & \geq \sum_{\ell=q_{1}n+1}^{q_{2}n}\frac{q_{1}}{\ell-1}\cdot\frac{1}{d}\OPT\left(\mathcal{I}_{heavy}\right)\\
 & = \OPT\left(\mathcal{I}_{heavy}\right)\frac{q_{1}}{d}\sum_{\ell=q_{1}n}^{q_{2}n-1}\frac{1}{\ell}\\
 & \geq \OPT\left(\mathcal{I}_{heavy}\right)\frac{q_{1}}{d}\ln\left(\frac{q_{2}}{q_{1}}\right).
\end{align*}
The first inequality follows from Lemma~\ref{Heavy phase lemma} and the second inequality is due to the fact that $\sum_{\ell=q_{1}n}^{q_{2}n-1}\frac{1}{\ell}\geq\intop_{q_{1}n}^{q_{2}n}\frac{1}{x}dx=\ln\left(\frac{q_{2}}{q_{1}}\right)$.
For the light phase we have
\begin{align*}
\sum_{\ell=q_{2}n+1}^{n}\E\left[R_{\ell}\right] & \geq \sum_{\ell=q_{2}n+1}^{n}\frac{q_{1}}{q_{2}}\left(1-2d\sum_{k=q_{2}n+1}^{\ell-1}\frac{1}{k}\right)\frac{1}{n}\OPT\left(\mathcal{I}_{light}\right)\\
 & = \frac{1}{n}\OPT\left(\mathcal{I}_{light}\right)\frac{q_{1}}{q_{2}}\left(\left(1-q_{2}\right)n-2d\sum_{k=q_{2}n+1}^{n}\left(\frac{n}{k}-1\right)\right)\\
 & \geq \OPT\left(\mathcal{I}_{light}\right)\frac{q_{1}}{q_{2}}\left(\left(2d+1\right)\left(1-q_{2}\right)-2d\ln\left(\frac{1}{q_{2}}\right)\right).
\end{align*}
The first inequality is due to Lemma~\ref{Light phase lemma} and the second inequality follows from the fact that $\sum_{k=q_{2}n+1}^{n} \frac{1}{k} \leq \intop_{q_{2}n}^{n}\frac{1}{x}dx = \ln\left(\frac{1}{q_{2}}\right)$. Overall we get
\begin{align}
\E\left[\ALG\right]\geq\OPT\left(\mathcal{I}_{heavy}\right)\frac{q_{1}}{d}\ln\left(\frac{q_{2}}{q_{1}}\right)+\OPT\left(\mathcal{I}_{light}\right)\frac{q_{1}}{q_{2}}\left(\left(2d+1\right)\left(1-q_{2}\right)-2d\ln\left(\frac{1}{q_{2}}\right)\right). \label{eq:VGAP_CR}
\end{align}
For the parameters $q_2 = 2d/\left(2d+1\right)$, $q_1 = q_2/\sqrt[4]{e}$, we have 
\[
\frac{q_1}{d} \ln \left( \frac{q_2}{q_1} \right) =  \frac{1}{4\sqrt[4]{e}} \frac{2}{2d+1} = \frac{1}{\sqrt[4]{e}(4d + 2)}.
\]
Using the fact that for $x \geq 0$, $\ln\left(1+x\right)\leq x-\frac{1}{2}x^{2}+\frac{1}{3}x^{3}$, we have
\[
\frac{q_{1}}{q_{2}}\left(\left(2d+1\right)\left(1-q_{2}\right)-2d\ln\left(\frac{1}{q_{2}}\right)\right) =\frac{1}{\sqrt[4]{e}} \left( 1 - 2d \ln\left({1+\frac{1}{2d}}\right) \right)
\geq \frac{1}{\sqrt[4]{e}} \left( \frac{1}{4d} - \frac{1}{12d^{2}} \right).
\]
It can be easily verified that $\left( 1/4d -1/12d^{2} \right) \geq 1/(4d+2)$ for $d \geq 1$. Now since $\OPT\left(\mathcal{I}_{heavy}\right)+\OPT\left(\mathcal{I}_{light}\right)\geq\OPT\left(\mathcal{I}\right)$, we get 
\begin{align*}
\E\left[\ALG\right] & \geq \frac{1}{\sqrt[4]{e}(4d + 2)} \OPT\left(\mathcal{I}_{heavy}\right) + \frac{1}{\sqrt[4]{e}(4d + 2)}\OPT\left(\mathcal{I}_{light}\right)\\
 & \geq \frac{1}{\sqrt[4]{e}(4d + 2)} \left( \OPT\left(\mathcal{I}_{heavy}\right) + \OPT\left(\mathcal{I}_{light}\right) \right) \\
 & \geq \frac{1}{\sqrt[4]{e}(4d + 2)} \OPT. \qedhere
\end{align*}

\end{proof}

It is important to note that for $d=1$, the competitive-ratio can be improved by choosing $q_1=0.5256$ and $q_2=0.69$. Setting these parameters in Inequality \eqref{eq:VGAP_CR} shows that Algorithm~\ref{VGap Algorithm} is $6.99$-competitive for the (one-dimensional) generalized assignment problem, which improves upon the best-known competitive-ratio of $8.1$ achieved by Kesselheim et al.~\cite{DBLP:journals/siamcomp/KesselheimRTV18}.
\begin{rem} \label{packing lps remark}
Algorithm~\ref{VGap Algorithm} can easily be extended to the case where each item has $K\geq1$ different packing options in each bin in a similar way to the algorithm of Kesselheim et al.~\cite{DBLP:journals/siamcomp/KesselheimRTV18}. Therefore, the general online packing LPs problem with $n$ requests and $m$ resources can be viewed as a special case of VGAP with one $m$-dimensional bin and $n$ items.
\end{rem}

\subsection{ The \{0,1\}-VGAP}
The \emph{$\{0,1\}$-VGAP} is a special case of VGAP in which the consumption of item $i$ from bin $j$ in dimension $t$ is either $0$ or the whole capacity of bin $j$ in dimension $t$. By scaling, we can assume without loss of generality that $\mathbf{b}_{j}=\mathbf{1}$, $\forall j \in [m]$ and $\mathbf{w}_{i,j}\in\left\{ 0,1\right\}^{d}$, $\forall i\in\left[n\right]$, $\forall j\in\left[m\right]$.\footnote{$\mathbf{1}$ denotes the all 1's vector.} Note that for $d=1$ the problem is identical to weighted bipartite matching.

As for the general VGAP, given an instance $\mathcal{I}$ we partition it into two sub-instances, however, we make the partition according to the density of the weight vectors. We call the packing option of item $i$ in bin $j$ \emph{dense} if $\left|\supp\left(\mathbf{w}_{i,j}\right)\right|\geq\sqrt{d}$,  otherwise, we call it \emph{sparse}.\footnote{$\supp(\cdot)$ denotes the set of indices of non-zero entries of a vector.} We denote by $\mathcal{I}_{dense}$ the sub-instance that consists only of the dense packing options of every item, and by 
$\mathcal{I}_{sparse}$ the complementary sub-instance that consists only of the sparse packing options.

Our algorithm for this case is similar to Algorithm~\ref{VGap Algorithm}. It is based on the simple observation that in $\mathcal{I}_{dense}$ at most $\sqrt{d}$ items can be packed in every bin, therefore, a maximum weight matching in $G(\mathcal{I}_{dense})$ has a weight of at least $\OPT({\mathcal{I}_{dense}})/\sqrt{d}$. As Algorithm~\ref{VGap Algorithm}, it operates in three phases: a sampling phase, the \emph{dense phase} in which it considers only dense options,
and the \emph{sparse phase} in which it considers only sparse options. The parameters $q_1$ and $q_2$ will be defined in the analysis.

\IncMargin{1em}
\begin{algorithm}
\label{0_1 VGAP Algorithm}
\caption{Online $\{0,1\}$-VGAP}

$S{}_{0}\leftarrow\emptyset$, $\PACK_{0}\leftarrow\emptyset$\;

\For { each item $i_{\ell}$ that arrives at round $\ell$ } {
    $ S_{\ell} \leftarrow S_{\ell-1} \cup \left\{ i_{\ell} \right\} $\;
    \uIf (\tcc*[f]{sampling phase}) {$\ell\leq q_{1}n$} {
		continue to the next round\;
    }
    \uElseIf( \tcc*[f]{dense phase}){ $q_{1}n+1\le\ell\leq q_{2}n$ } {
        Let $x^{\left(\ell\right)}$ be a maximum-weight matching in
        $G\left(\mathcal{I}_{dense}|_{S_{\ell}}\right)$;

		\tcp*[h]{compute a tentative assignment $\left(i_{\ell},j_{\ell}\right)$}

		\uIf {$i_\ell$ is matched in $x^{\left(\ell\right)}$} {
			Let $j_\ell$ be the bin to which $i_\ell$ is matched\; 
		} 
		\Else {
			$j_\ell \leftarrow 0 $
		}
		\If {$j_\ell\neq0$ {\bf and} $j_\ell$ is empty in $\PACK_{\ell - 1}$} {
			$\PACK_{\ell}\leftarrow \PACK_{\ell-1}\cup\left\{ \left(i_{\ell},j_{\ell}\right)\right\}$\;
    	}
    }
    \Else( \tcp*[h] {({$\ell \geq q_{2}n + 1$})} \tcc*[f]{ sparse phase  }) { 
        Let $x^{\left(\ell\right)}$ be an optimal fractional solution for
        the LP-relaxation of $\mathcal{I}_{sparse}|_{S_{\ell}}$\;

		\tcp*[h]{compute a tentative assignment $\left(i_{\ell},j_{\ell}\right)$ by randomized rounding}

		Choose bin $j_{\ell}$ randomly where $\Pr\left[j_{\ell}=j\right]=x_{i_{\ell},j}^{\left(\ell\right)}$
    	and $\Pr\left[j_{\ell}=0\right]=1-{\displaystyle \sum_{j\in\left[m\right]}}x_{i_{\ell},j}^{\left(\ell\right)}$\;
   	 \If{$j_\ell\neq0$ {\bf and} $ \PACK_{\ell - 1} \cup \left\{\left(i_{\ell},j_{\ell}\right)\right\}$ is feasible} {
        	$\PACK_{\ell}\leftarrow \PACK_{\ell-1}\cup\left\{ \left(i_{\ell},j_{\ell}\right)\right\}$\;
    	}
    }
\Return{ $\PACK_{n}$ }
}
\end{algorithm}
\DecMargin{1em}
\begin{lem}
\label{Dense phase lemma}
For $ q_{1}n + 1 \leq \ell \leq q_{2}n$, we have 
\[
\E[R_{\ell}] \geq \frac {q_{1}}{\ell - 1} \cdot \frac {1}{\sqrt{d}} \OPT(\mathcal{I}_{dense}).
\]
\end{lem} 
The proof is similar to the proof of Lemma \ref{Heavy phase lemma} by using the observation above. 

\begin{lem}
\label{Sparse phase lemma}
For $ q_{1}n + 1 \leq \ell \leq q_{2}n$, we have
\[
\E\left[R_{\ell}\right] \geq \frac{q_{1}}{q_{2}}\left(1-\sqrt{d}\sum_{k=q_{2}n+1}^{\ell-1}\frac{1}{k}\right) \frac{1}{n}\OPT\left(\mathcal{I}_{sparse}\right).
\]
\end{lem}
\begin{proof}
As in Lemma~\ref{Light phase lemma}, we have $\E\left[p_{i_{\ell},j_{\ell}}\right]=\frac{1}{n}\OPT\left(\mathcal{I}_{sparse}\right)$
where the expectation is taken only over the random choice of the subset $S_{\ell}\subseteq\left[n\right]$, the random choice of $i_{\ell}\in S_{\ell}$ and the internal randomness of the algorithm at round $\ell$. Once again we bound the probability of successful assignment over the random arrival order of previous items and the internal randomness of the algorithm in previous rounds. The assignment of $i_{\ell}$ to $j_{\ell}$ must be successful if the following conditions hold: (1) no item was packed in $j_{\ell}$ during the dense phase, and (2) no tentative assignments from previous rounds of the sparse phase occupy the entries in $\supp\left(\mathbf{w}_{i_{\ell}, j_{\ell}}\right)$ of $j_{\ell}$. Let us denote event (1) by $H_\ell$ and the event described in (2) by $L_{\ell}$. At round $q_{2}n \leq k \leq \ell$ the algorithm uses an optimal fractional solution $x^{\left(k\right)}$ for the LP-relaxation on $\mathcal{I}|_{S_{k}}$ to compute a tentative assignment $\left(i_{k},j_{k}\right)$. In that solution we have $\sum_{i\in S_{k}}x_{i,j_{\ell}}^{\left(k\right)}w_{i,j_{\ell}}^{t}\leq1$, $\forall t\in\left[d\right]$.
Observe that by the randomized rounding at round $k$ and the fact that $w_{i_{k},j_{\ell}}^{t} \in \left\{ 0,1 \right\}$, the probability that the tentative assignment of $i_{k}$ uses dimension $t$ in $j_{\ell}$ is given by
\[
\sum_{i \in S_k} \Pr[ j_{k} = j_{\ell} \wedge w_{i,j_k}^{t} = 1 |i_k = i]\cdot\Pr[i_{k}=i] = 
 \frac{1}{k} \sum_{i \in S_k} x_{i,j_{\ell}}^{\left(k\right)} w_{i,j_{\ell}}^{t} \leq \frac{1}{k}.
\]
Using a union bound, since $\left|\supp\left(\mathbf{w}_{i_{\ell},j_{\ell}}\right) \right| \leq \sqrt{d}$, the probability that $i_{k}$ blocks $i_{\ell}$ from being packed is at most $\sqrt{d}/k$. Applying a union bound once again over all previous rounds of the sparse phase, we get
\[
\Pr[L_\ell]\geq 1-\sum_{k=q_2n + 1}^{\ell -1} {\frac{\sqrt{d}}{k}}.
\]
From here on we can follow a similar argument as in the proof of Lemma~\ref{Light phase lemma} and get
\begin{align*}
\Pr \left[ \text{successful assignment} \right] & \geq 
\Pr \left[H_{\ell} \wedge L_{\ell} \right] \\
& \geq \prod_{k=q_{1}n+1}^{q_{2}n}\left(1-\frac{1}{k}\right) \left( 1-\sum_{k=q_2n + 1}^{\ell -1} {\frac{\sqrt{d}}{k}} \right)\\
& = \frac{q_{1}}{q_{2}}\left(1-\sqrt{d}\sum_{k=q_{2}n+1}^{\ell-1}\frac{1}{k}\right).
\end{align*}
Overall, we get the lemma. 
\end{proof}

\begin{thm}
\label{Therorem_0_1gap_competitive} For $q_2 = \sqrt{d}/\left(\sqrt{d}+1\right)$, $q_1 = q_2/\sqrt{e}$ Algorithm~\ref{0_1 VGAP Algorithm} is $2\sqrt{e}\left(\sqrt{d}+2\right)$-competitive.
\end{thm}

\begin{proof}
As in the proof of Theorem~\ref{Therorem_gap_competitive}, by using Lemma~\ref{Dense phase lemma} and Lemma~\ref{Sparse phase lemma} to sum over the profit raised in each phase, we get
\[
\E\left[\ALG\right]\geq\OPT\left(\mathcal{I}_{dense}\right)\frac{q_{1}}{\sqrt{d}}\ln\left(\frac{q_{2}}{q_{1}}\right)+
\OPT\left(\mathcal{I}_{sparse}\right)\frac{q_{1}}{q_{2}}\left(\left(\sqrt{d}+1\right)\left(1-q_{2}\right)-\sqrt{d}\ln\left(\frac{1}{q_{2}}\right)\right).
\]
Setting $q_{2}=\frac{\sqrt{d}}{\sqrt{d}+ 1}, q_{1}=q_{2}/\sqrt{e}$, we get
\begin{align*}
\E\left[\ALG\right] &\geq   \OPT\left(\mathcal{I}_{dense}\right) \frac{1}{2\sqrt{e}\left( \sqrt{d}+1 \right)} +  \OPT\left(\mathcal{I}_{sparse}\right)\frac{1}{\sqrt{e}} \left( 1-\sqrt{d}\ln\left( 1+ \frac{1}{\sqrt{d}} \right) \right).
\end{align*}
To bound the second term we use the fact that for $x\geq0$, $\ln(1+x) \leq x- \frac{1}{2}x^2 + \frac{1}{3}x^3$ and get
\begin{align*}
 \frac{1}{\sqrt{e}} \left( 1-\sqrt{d}\ln\left( 1+ \frac{1}{\sqrt{d}} \right) \right)
\geq \frac{1}{\sqrt{e}} \left( \frac{1}{2\sqrt{d}} - \frac{1}{3d} \right) 
\geq \frac{1}{\sqrt{e}} \left( \frac{1}{2\sqrt{d}+4}\right).
\end{align*}
The second inequality can be easily verified to holds for $d\geq 1$.
Using it and the fact that $\OPT\left(\mathcal{I}_{dense}\right)+\OPT\left(\mathcal{I}_{sparse}\right)\geq\OPT\left(\mathcal{I}\right)$, we get
\begin{align*}
\E\left[\ALG\right] &\geq  \OPT\left(\mathcal{I}_{dense}\right) \frac{1}{2\sqrt{e}\left( \sqrt{d}+1 \right)} +  \OPT\left(\mathcal{I}_{sparse}\right)\frac{1}{2\sqrt{e}\left(\sqrt{d}+2\right)} \geq  \frac{1}{2\sqrt{e}\left(\sqrt{d}+2\right)} \OPT. \qedhere
\end{align*}
\end{proof}

\section{Vector Multiple Knapsack Problem }

The \emph{Vector Multiple Knapsack Problem} (VMKP) is a special
case of VGAP in which all bins have a capacity of $\mathbf{1}$,
every packing option of item $i$ consumes the same amount of capacity
$\mathbf{w}_{i}\in \left[ 0,1\right]^{d}$
and provides the same profit $p_{i}\geq0$, i.e., $\mathbf{w}_{i,j}=\mathbf{w}_{i}$,
$p_{i,j}=p_{i}$, $\forall i\in\left[n\right]$, $\forall j\in\left[m\right]$.
We study the case where there are at least two bins, i.e., $m\geq2$.
For this special case we present an online algorithm that improves
upon the competitive-ratio of Algorithm~\ref{VGap Algorithm}. 

\IncMargin{1em}
\begin{algorithm}
\label{MKS Algorithm}
\SetKwInOut{Input}{input}\SetKwInOut{Output}{output}
\caption{Online VMKP}

$S{}_{0}\leftarrow\emptyset$, $\PACK_{0}\leftarrow\emptyset$\;

\For{each item $i_{\ell}$ that arrives at round $\ell$} {
    $S_{\ell}\leftarrow S_{\ell-1}\cup\left\{ i_{\ell}\right\} $\;
    \uIf(\tcc*[f]{sampling phase}){$\ell\leq qn$} {
		continue to the next round\;
    }
    \Else( \tcp*[h]{ {$\ell \geq qn + 1$}  } \tcc*[f]{packing phase} ) {

        Let $x^{\left(\ell\right)}$ be an optimal fractional solution for the LP-relaxation of $\mathcal{I}|_{S_{\ell}}$;

		Choose $j_{\ell}$ randomly where $\Pr\left[j_{\ell}=j\right]=x_{i_{\ell},j}^{\left(\ell\right)}$ and $\Pr\left[j_{\ell}=0\right]=1-{\displaystyle \sum_{j\in\left[m\right]}}x_{i_{\ell},j}^{\left(\ell\right)}$\;

		\tcp*[h]{First Fit}

		Let $B_{\ell} = \left\{ j\in[m] : \PACK_{\ell} \cup \left\{\left(i_\ell,j\right)\right\} \text{is feasible} \right\}$\;
		\uIf{ $j_\ell \neq 0$ {\bf and} $B_{\ell} \neq \emptyset$ } {
			$ \PACK_{\ell} \leftarrow \PACK_{\ell - 1}\cup\left\{ \left(i_{\ell},\min{B_{\ell}} \right) \right\}$\;
		}
	}
}
\Return{ $\PACK_{n}$ }
\end{algorithm}
\DecMargin{1em}

Algorithm~\ref{MKS Algorithm} consists of two phases: a sampling
phase and a packing phase. The packing phase is similar to the light
phase of Algorithm~\ref{VGap Algorithm}, however, instead of using
the LP-solution to compute a tentative assignment, it uses it only
to make a binary decision whether to pack the current item or not.
For the actual packing, it exploits the fact that all packing options
are identical and uses the First Fit algorithm~\cite{garey1976resource}.

We now analyze the performance of Algorithm~\ref{MKS Algorithm}.
First we prove a simple observation due to the nature of First Fit. 
\begin{lem}
\label{First Fit lemma}For $\ell\geq qn+1$ and $m\geq2$, if $i_{\ell}$
cannot be packed in any bin, then
\[
\sum_{\left(i,j\right)\in P_{\ell-1}}\sum_{t=1}^{d}w_{i}^{t}\geq \frac{m}{2}.
\]
\end{lem}
\begin{proof}
Let $u\left(j,t,\ell\right)$ denote the total consumption of bin
$j$ in dimension $t$ before round $\ell$. Since $i_{\ell}$ cannot
be packed in any bin, there is at least one item packed in each bin.
Consider any two bins $j'>j$, and let $i_{k}$ be the first item
that was packed in $j'$. $i_{k}$ could not be packed in bin $j$,
therefore, for some $t'\in\left[d\right]$ we have $u\left(j,t',k\right)+w_{i_{k}}^{t'}>1$.
Since the consumption is non-decreasing and $u\left(j',t',\ell\right)\geq w_{i_{k}}^{t'}$
we have $u\left(j,t',\ell\right)+u\left(j',t',\ell\right)>1$, therefore,
$\sum_{t=1}^{d}u\left(j,t,\ell\right)+u\left(j',t,\ell\right)>1$.
By summing the last inequality for all consecutive pairs of bins $\left(j+1,j\right)$
as well as $\left(m,1\right)$ we get $2\sum_{j=1}^{m}\sum_{t=1}^{d}u\left(j,t,\ell\right)>m$
and hence the lemma. 
\end{proof}
Next, we follow the method of the previous section to bound the expected
profit of the algorithm at each round.
\begin{lem}
\label{MKP_profit} For $\ell\geq qn+1$ and $m\geq2$, we have 
\[
\E\left[R_{\ell}\right]\geq\left(1-2d\sum_{k=qn+1}^{\ell-1}\frac{1}{k}\right)\frac{1}{n}\OPT.
\]
\end{lem}

\begin{proof}
By following a similar argument to that in Lemma~\ref{Heavy phase lemma},
we get $\E\left[p_{i_{\ell},j_{\ell}}\right]\geq\frac{1}{n}\OPT$
for $\ell\geq qn+1$. We now bound the probability that $\sum_{\left(i,j\right)\in P_{\ell-1}}\sum_{t=1}^{d}w_{i}^{t}<m/2$, by Lemma~\ref{First Fit lemma}, this is a sufficient condition for the  assignment of
$i_{\ell}$ to be successful. At round $k<\ell$ the algorithm computes
a tentative assignment based on an optimal fractional solution $x^{\left(k\right)}$
for the LP-relaxation of $\mathcal{I}|_{S_{k}}$, therefore, we have
$\sum_{i\in S_{k}}\sum_{t=1}^{d}\sum_{j=1}^{m}x_{i,j}^{\left(k\right)}w_{i}^{t} \leq d m$.
Since $i_{k}$ is a uniformly random item of $S_{k}$, we have $\E\left[\sum_{t=1}^{d}\sum_{j=1}^{m}x_{i_{k},j}^{\left(k\right)}w_{i_{k}}^{t}\right]\leq {d m}/{k}$,
hence, 
\[
\E\left[\sum_{\left(i,j\right)\in P_{\ell-1}}\sum_{t=1}^{d}w_{i}^{t}\right]
\leq \E\left[\sum_{k=qn+1}^{\ell-1}\sum_{t=1}^{d}w_{i_{k}}^{t}\sum_{j=1}^{m}{x_{i_{k},j}^{\left(k\right)}}\right]
= \sum_{k=qn+1}^{\ell-1}\E\left[\sum_{t=1}^{d}\sum_{j=1}^{m}{x_{i_{k},j}^{\left(k\right)}w_{i_{k}}^{t}}\right]
\leq\sum_{k=qn+1}^{\ell-1}\frac{d m}{k}.
\]
As before, we can now use Markov's inequality to bound the probability
of successful assignment and get the lemma.
\end{proof}

\begin{thm}
For $q=2d/(2d+1)$, Algorithm~\ref{MKS Algorithm} is $\left( 4d + 2 \right)$-competitive.
\end{thm}

\begin{proof}
By Lemma~\ref{MKP_profit}, the overall profit of the algorithm is
bounded by 
\begin{align*}
\E\left[\text{ALG}\right] & = \sum_{\ell=qn+1}^{n}\E\left[R_{\ell}\right]\\
 & \geq  \sum_{\ell=qn+1}^{n}\left(1-2d\sum_{k=qn+1}^{\ell-1}\frac{1}{k}\right)\frac{1}{n}\OPT\\
 & \geq  \left(\left(2d+1\right)\left(1-q\right)-2d\ln\left(\frac{1}{q}\right)\right)\OPT.
\end{align*}
This bound is maximized for $q=2d/\left(2d+1\right)$, and for this choice of parameter, using similar arguments as in the proof of Theorem~\ref{Therorem_gap_competitive}, we get
\begin{align}
\E\left[\text{ALG}\right] &\geq \left(1-2d\ln\left(1+\frac{1}{2d}\right)\right)\OPT \label{eqn:VMK_competitive_ratio}\\
 &\geq  \left(\frac{1}{4d} + \frac{1}{12d^{2}} \right) \OPT \notag\\
 &\geq  \left( \frac{1}{4d + 2} \right)\OPT.\qedhere \notag
\end{align}
\end{proof}

Note that by setting $d=1$ in Inequality~\eqref{eqn:VMK_competitive_ratio} we get that Algorithm~\ref{MKS Algorithm} is $5.29$-competitive for the (one-dimensional) multiple knapsack problem with at least two bins.
\begin{rem}
For the special case of $d=1$, Algorithm~\ref{MKS Algorithm} can
be implemented in a more efficient way: instead of solving an LP-relaxation
at every round of the packing phase, we can obtain an optimal fractional
solution by using a simple greedy algorithm. 
\end{rem}
\begin{rem}
Algorithm~\ref{MKS Algorithm} can be extended to the case of variable-sized squared bins, that it, to the case where $\mathbf{b}_{j}=\mathbf{1}\cdot{b}_{j}$, $\forall j \in \left[m\right]$, under the assumption that every item fits into every bin, i.e., $w_{i}^{t}\leq b_{j}$ $\forall i\in\left[n\right],\forall j\in\left[m\right],\forall t\in\left[d\right]$,
 through sorting the bins by their capacity in a non-increasing
order.
\end{rem}

\section{Lower bound}
We now prove a lower bound of $\Omega(d)$ for the vector knapsack problem (VMKP with a single bin). Since it is a special case of VGAP, it shows our $O(d)$-competitive algorithm for VGAP from Section~\ref{VGAP section} is asymptotically optimal. Note that our lower bound holds without any complexity assumptions. In particular, it also applies to algorithms with unbounded computational power. The proof is inspired by the work of Babaioff et al.~\cite{DBLP:conf/soda/BabaioffIK07}.

We construct an instance for the $d$-dimensional knapsack problem consisting of one bin of capacity $\mathbf{1}$ and $n= \delta d^{\left(\delta + 1\right)d+ 1}$ items, where $\delta\in\mathbb{N}_{+}$. The weight vectors of the items are the columns of the following $d \times d$ matrices: 
\[
A_j = \left(1 - \epsilon j d^j\right) \cdot I + \epsilon j d^{j-1} \cdot \left(\mathbf{1}\mathbf{1}^{T} - I\right), \quad  \forall j\in[\delta d^{(\delta + 1)d}].
\]
Where $I$ is the $d \times d$ identity matrix, and $\epsilon < 1/ \left( 2nd^{n} \right) $. By the choice of $\epsilon$ it holds that $\epsilon j  d^{j} < 1/2$, $\forall j \in [\delta d^{(\delta + 1)d}]$.

Observe that for every matrix $A_j$, all the items that correspond to its columns fit together in the bin, that is, $A_j \cdot \mathbf{1}  \leq \mathbf{1}$. Also, every two columns of different matrices cannot be packed together. This is true because for any two matrices $A_i,A_j$ where $i>j$, and any two columns $k,\ell \in [d]$, we have
\[
 \left(A_j\right)_{k,k} + \left(A_i\right)_{k,\ell}  \geq \left(1 - \epsilon j d^{j}\right) +  \epsilon i  d^{i-1}  \geq 1 - \epsilon j d^{j} +  \epsilon \left(j+1\right)  d^{j} > 1.
\]
The first inequality follows from the fact that $\epsilon i  d^{i-1} < \left(1 - \epsilon i d^i\right) $, and the second inequality follows from the fact that $i \geq j+1$. Every item is independently assigned a profit of $1$ with probability $1/d^{\delta+1}$ and $0$ with probability $1-1/d^{\delta + 1}$.

\begin{thm} \label{lowerbound thm}
Any online algorithm produces a packing with expected profit of at most $\left (1 + \frac{1}{d^{\delta}} \right)$, while $\OPT=d$ with probability of at least $\left(1-\frac{1}{e^\delta}\right)$.
\end{thm}

\begin{proof}
Let us observe the first item that the online algorithm packs. It corresponds to a column of one matrix $A_j$. All items that correspond to columns of different matrices cannot be packed along with it. The only items that can be added to the packing are the remaining columns of $A_j$. There are less than $d$ such items left, each has an expected profit of $1/d^{\delta+1}$. Since the first item has a profit of at most $1$, the expected profit of the packing produced by the algorithm is at most $1+1/d^{\delta}$.

With regard to the optimal packing, for a given matrix $A_\ell$, the probability that all items are of profit $1$ is $1/d^{(\delta + 1)d}$, therefore, the probability that all matrices are of weight less than $d$ is
\[
\left(1-\frac{1}{d^{(\delta + 1)d}}\right)^{d^{\left(\delta + 1\right)d}\delta} \leq \frac{1}{e^{\delta}}. \qedhere
\]
\end{proof}

Note that Theorem~\ref{lowerbound thm} can be easily modified to apply to the case of two identical bins, thus, it shows that Algorithm~\ref{MKS Algorithm} for VMKP with at least two bins is also asymptotically optimal.

\section{Conclusions} 
In this paper, we presented simple, asymptotically optimal, online algorithms for multidimensional variants of the generalized assignment problem in the random-order model, which has vast implications for real-world applications, like resource allocation in cloud computing.  

Our bounds for VGAP are translated to a matching lower and upper bounds of $\Omega(m)$ and $O(m)$ for the general online packing LPs problem (as mentioned in Remark~\ref{packing lps remark}, where $m$ is the number of resources). 

For the one-dimensional case, the best lower bound for the online GAP is derived from the lower bound for the secretary problem of $e$~\cite{DBLP:journals/mor/BuchbinderJS14}. An interesting open question is to close the gap between $e$ and the upper bound of $6.99$ presented in this paper. It is also very interesting to understand whether the new theoretical algorithm provides practical value for cloud resource allocation, where the value of $d$ is a small constant ($2$ or $3$).

\bibliographystyle{plain}
\bibliography{citations}

\end{document}